\documentclass[11pt]{article}

\usepackage{amsthm}
\usepackage{color}
\usepackage[left=1in,top=1in,right=1in]{geometry}
\usepackage{epsfig,graphicx,amssymb,amsmath,amsfonts,mathrsfs}
\usepackage{sidecap}
\usepackage{url}

\usepackage{hyperref} 

\newtheorem{lemma}{Lemma}
\newtheorem{thm}{Theorem}

\newtheorem{cor}{Corollary}

\newcommand{\hide}[1]{}

\def\ie{{\it i.e.}}
\def\eg{{\it e.g.}}

\def\etal{{\it et al.}}

\def\eps{\varepsilon}
\def\opt{\mathrm{OPT}}

\newcommand{\ceil}[1]{\left\lceil #1 \right\rceil}

\newcommand{\RR}{\mathbb{R}}
\newcommand{\UUU}{\mathcal{U}}
\newcommand{\SSS}{\mathcal{S}}
\newcommand{\HHH}{\mathcal{H}}
\newcommand{\setcover}{{\sc Set Cover}}
\newcommand{\hittingset}{{\sc Hitting Set}}

\begin{document}

\title{Improved Approximation for Guarding Simple Galleries from the Perimeter\thanks{Some of these results appeared in preliminary form as \emph{D. Kirkpatrick. Guarding galleries with no nooks. In Proceedings of the 12th Canadian Conference on Computational Geometry (CCCG'00), pages 43--46, 2000.}}}
\author{
James King\\
\small{School of Computer Science}\\
\small{McGill University}\\
\small{\texttt{jking@cs.mcgill.ca}}
\and
David Kirkpatrick\\
\small{Department of Computer Science}\\
\small{University of British Columbia}\\
\small{\texttt{kirk@cs.ubc.ca}}
}

\maketitle


\begin{abstract}
We provide an $O(\log \log \opt)$-approximation algorithm for the problem of guarding a simple polygon with guards on the perimeter.  We first design a polynomial-time algorithm for building $\eps$-nets of size $O\kern-2pt\left(\frac{1}{\eps}\log \log \frac{1}{\eps}\right)$ for the instances of \hittingset{} associated with our guarding problem.  We then apply the technique of Br\"{o}nnimann and Goodrich to build an approximation algorithm from this $\eps$-net finder.  Along with a simple polygon $P$, our algorithm takes as input a finite set of potential guard locations that must include the polygon's vertices.  If a finite set of potential guard locations is not specified, \eg{}\null{} when guards may be placed anywhere on the perimeter, we use a known discretization technique at the cost of making the algorithm's running time potentially linear in the ratio between the longest and shortest distances between vertices.  Our algorithm is the first to improve upon $O(\log\opt)$-approximation algorithms that use generic net finders for set systems of finite VC-dimension.
\end{abstract}

\section{Introduction}

\subsection{The art gallery problem}

In computational geometry, art gallery problems are motivated by the question, ``How many security cameras are required to guard an art gallery?''  The art gallery is modeled as a connected polygon $P$.  A camera, which we will henceforth call a \emph{guard}, is modeled as a point in the polygon, and we say that a guard $g$ \emph{sees} a point $q$ in the polygon if the line segment $\overline{gq}$ is contained in $P$.  We call a set $G$ of points a \emph{guarding set} if every point in $P$ is seen by some $g\in G$.  Let $V(P)$ denote the vertex set of $P$ and let $\partial P$ denote the boundary of $P$.  We assume that $P$ is closed and non-degenerate so that $V(P) \subset \partial P \subset P$.

We consider the minimization problem that asks, given an input polygon $P$ with $n$ vertices, for a minimum guarding set for $P$.  Variants of this problem typically differ based on what points in $P$ must be guarded and where guards can be placed, as well as whether $P$ is simple or contains holes.  Typically we want to guard either $P$ or $\partial P$, and our set of potential guards is typically $V(P)$ (vertex guards), $\partial P$ (perimeter guards), or $P$ (point guards).  For results on art gallery problems not related to minimization problems we direct the reader to O'Rourke's book \cite{orourke87}, which is available for free online.

The problem was proved to be NP-complete first for polygons with holes by O'Rourke and Supowit \cite{orourke83}.  For guarding simple polygons it was proved to be NP-complete for vertex guards by Lee and Lin \cite{lee86}; their proof was generalized to work for point guards by Aggarwal \cite{aggarwal84}.  This raises the question of approximability.  There are two major hardness results.  First, for guarding simple polygons, Eidenbenz \cite{eidenbenz98} proved that the problem is APX-complete, meaning that we cannot do better than a constant-factor approximation algorithm unless ${\rm P}={\rm NP}$.  Subsequently, for guarding polygons with holes, Eidenbenz \etal{}\null{} \cite{Eidenbenz01} proved that the minimization problem is as hard to approximate as \setcover{} in general if there is no restriction on the number of holes.  It therefore follows from results about the inapproximability of \setcover{} by Feige \cite{Feige98} and Raz and Safra \cite{Raz97} that, for polygons with holes, it is NP-hard to find a guarding set of size $o(\log n)$.  These hardness results hold whether we are dealing with vertex guards, perimeter guards, or point guards.

Ghosh \cite{ghosh87} provided an $O(\log n)$-approximation algorithm for guarding polygons with or without holes with vertex guards.  His algorithm decomposes the input polygon into a polynomial number of cells such that each point in a given cell is seen by the same set of vertices.  This discretization allows the guarding problem to be treated as an instance of \setcover{} and solved using general techniques.  This will be discussed further in Section \ref{sec:setcover}.  In fact, applying methods for \setcover{} developed after Ghosh's algorithm, it is easy to obtain an approximation factor of $O(\log\opt)$ for vertex guarding simple polygons or $O(\log h \log OPT)$ for vertex guarding a polygon with $h$ holes.  

When considering point guards or perimeter guards, discretization is far more complicated since two distinct points will not typically be seen by the same set of potential guards even if they are very close to each other.  Deshpande \etal{}\null{} \cite{deshpande07} obtain an approximation factor of $O(\log\opt)$ for point guards or perimeter guards by developing a sophisticated discretization method that runs in pseudopolynomial time\footnote{It is a pseudopolynomial-time algorithm in that its running time may be linear in the ratio between the longest and shortest distances between two vertices.}.  Efrat and Har-Peled \cite{efrat06} provided a randomized algorithm with the same approximation ratio that runs in fully polynomial expected time; their discretization technique involves only considering guards that lie on the points of a very fine grid.

Our contribution is an algorithm for guarding simple polygons, using either vertex guards or perimeter guards.  Our algorithm has a guaranteed approximation factor of $O(\log\log\opt)$ and the running time is polynomial in $n$ and the number of potential guard locations.  This is the best approximation factor obtained for vertex guards and perimeter guards.  If no finite set of guard locations is given, we use the discretization technique of Deshpande \etal{}\null{} and our algorithm is polynomial in $n$ and $\Delta$, where $\Delta$ is the ratio between the longest and shortest distances between vertices.  

\subsection{Guarding problems as instances of \hittingset{}}\label{sec:setcover}

\subsubsection{Set Cover and Hitting Set}
\setcover~is a well-studied NP-complete optimization problem.  Given a universe $\mathcal{U}$ of elements and a collection $\mathcal{S}$ of subsets of $\mathcal{U}$, \setcover~asks for a minimum subset $\mathcal{C}$ of $\mathcal{S}$ such that $\bigcup_{S\in\mathcal{C}}S = \mathcal{U}$. In other words, we want to cover all of the elements in $\UUU$ with the minimum number of sets from $\mathcal{S}$.  In general, \setcover~is not only difficult to solve exactly (see, \eg, \cite{Garey79}) but is also difficult to approximate---no polynomial time approximation algorithm
can have a $o(\log n)$ approximation factor unless $\mathrm{P=NP}$ \cite{Raz97}.  Conversely, a simple greedy heuristic\footnote{The heuristic repeatedly picks the set that covers the most uncovered elements.} \cite{chvatal79} for \setcover~attains an $O(\log n)$ approximation factor.  Another problem, \hittingset, asks for a minimum subset $\mathcal{H}$ of $\mathcal{U}$ such that $S \bigcap \mathcal{H} \neq \emptyset$ for any $S\in\SSS$.  Any instance of \hittingset{} can easily be formulated as an instance of \setcover{} and vice versa.

\subsubsection{Set Systems of Guarding Problems}
Guarding problems can naturally be expressed as instances of \setcover{} or \hittingset.  We wish to model an instance of a guarding problem as an instance of \hittingset{}.  The desired set system $(\mathcal{U},\mathcal{S})$ is constructed as follows.  $\mathcal{U}$ contains the potential guard locations.  For each point $p$ that needs to be guarded, $S_p$ is the set of potential guards that see $p$, and $\SSS = \{S_p \;|\; p \in P \}$.

\subsubsection{$\eps$-Nets}
Informally, if we wish to relax the \hittingset{} problem, we can ask for a subset of $\mathcal{U}$ that hits all \emph{heavy} sets in $\mathcal{S}$.  This is the idea behind $\eps$-nets.  For a set system $(\mathcal{U},\mathcal{S})$ and an additive weight function $w$, an $\eps$-net is a subset of $\mathcal{U}$ that hits every set in $\mathcal{S}$ having weight at least $\eps\cdot w(\mathcal{U})$.

It is known that set systems of VC-dimension $d$ admit $\eps$-nets of size $O\kern-2pt\left(\frac{d}{\eps}\log \frac{1}{\eps}\right)$ \cite{blumer89} and that this is asymptotically optimal without further restrictions \cite{komlos92}.  It is also known that set systems associated with the guarding of simple polygons with point guards\footnote{This bound also applies \textit{a fortiori} to perimeter guards and vertex guards.} have constant VC-dimension \cite{kalai97,valtr98}.  Thus when guarding simple polygons we can construct $\eps$-nets of size $O\kern-2pt\left(\frac{1}{\eps}\log\frac{1}{\eps}\right)$ using general techniques.  In a polygon with $h$ holes the VC-dimension is $O(\log h)$ \cite{valtr98} and therefore $\eps$-nets of size $O\kern-2pt\left(\frac{1}{\eps}\log\frac{1}{\eps}\log h\right)$ can be constructed.

Using techniques specific to vertex guarding or perimeter guarding a simple polygon, we are able to break through the general $\Theta\kern-2pt\left(\frac{d}{\eps}\log \frac{1}{\eps}\right)$ lower bound to build smaller $\eps$-nets.  This result is stated in the following theorem.
\begin{thm}\label{thm:main}
For the problem of guarding a simple polygon with vertex guards or perimeter guards, we can build $\eps$-nets of size $O\kern-2pt\left(\frac{1}{\eps}\log \log \frac{1}{\eps}\right)$ in polynomial time.
\end{thm}
\begin{proof}
In Section \ref{sec:quadratic} we introduce the basic ideas that allow the construction of $\eps$-nets of size $O(1/\eps^2)$.  In Section \ref{sec:hierarchical} we give a more complicated, hierarchical technique that lets us construct $\eps$-nets of size 
$O\kern-2pt\left(\frac{1}{\eps}\log \log \frac{1}{\eps}\right)$.
\end{proof}

A similar result for a different problem was recently obtained by Aronov \etal{} \cite{aronov09}, who proved the existence of $\eps$-nets of size $O\kern-2pt\left(\frac{1}{\eps}\log \log \frac{1}{\eps}\right)$ when $\SSS$ is either a set of axis-parallel rectangles in $\RR^2$ or axis-parallel boxes in $\RR^3$.

\subsubsection{Approximating \hittingset{} with $\eps$-Nets}
Br\"{o}nnimann and Goodrich \cite{Bronnimann95} introduced an algorithm for using a \emph{net finder} (an algorithm for finding $\eps$-nets) to find approximately optimal solutions for the \hittingset{} problem.  Their algorithm gives weights (initially uniform\footnote{Initial weights can be non-uniform but this is not necessary for our purposes.}) to the elements in $\UUU$.  The net finder is then used to find an $\eps$-net for $\eps = 1/2c'$, with $c'$ fixed at a constant between 1 and $2\cdot \opt$.  If there is a set in $\SSS$ not hit by the $\eps$-net, the algorithm picks such a set and doubles the weight of every element in it.  It then repeats, finding a new $\eps$-net given the new weighting.  This continues until the algorithm finds an $\eps$-net that hits every set in $\SSS$.  If the net finder constructs $\eps$-nets of size $f(1/\eps)$, their main algorithm finds a hitting set of size $f(4\cdot\opt)$.

Previous approximation algorithms achieving guaranteed approximation factors of $\Theta(\log\opt)$ \cite{deshpande07,efrat06} have used this technique, along with generic $\eps$-net finders of size $O\kern-2pt\left(\frac{1}{\eps}\log\frac{1}{\eps}\right)$ for set systems of constant VC-dimension.  Instead, we use our net finder from Theorem \ref{thm:main} to obtain the following corollary, whose proof is given in Section \ref{sec:main}. 

\begin{cor}\label{cor:main}
Let $P$ be a simple polygon with $n$ vertices and let $G$ be a set of potential guard locations such that $V(P)\subseteq G \subset \partial P$.  Let $T\subseteq P$ be the set of points we want to guard.  There is a polynomial-time algorithm that outputs a guarding set for $T$ of size $O(\opt \cdot \log\log \opt)$, where $\opt$ is the size of the minimum subset of $G$ that guards $T$.
\end{cor}

\section{The Main Algorithm}\label{sec:main}

\subsection{Main algorithm.}
Our main algorithm is an application of that presented by Br\"{o}nnimann and Goodrich \cite{Bronnimann95}.  Their algorithm provides a generic way to turn a \emph{net finder}, \ie{} an algorithm for finding $\eps$-nets for an instance of \hittingset{}, into an approximation algorithm.  Along with a net finder we also need a \emph{verifier}, which either states correctly that a set $H$ is a hitting set, or returns a set from $\SSS$ that is not hit by $H$.

For the sake of completeness we present the entire algorithm here.  $G$ is the set of potential guard locations and $T$ is the set of points that must be guarded.  We first assign a weight function $w$ to the set $G$.  When the algorithm starts each element of $G$ has weight 1.  The main idea of the algorithm is to repeatedly find an $\eps$-net $H$ and, if $H$ is not a hitting set (\ie{}\null{} if it does not see everything in $T$), to choose a point $p \in T$ that is not seen by $H$ and double the weight of any guard that sees $p$.

\subsubsection{Bounding the number of iterations.}
For now assume we know the value of $\opt$ and we set $\eps = \frac{1}{2\cdot\opt}$.  We give an upper bound for the number of doubling iterations the algorithm can perform.  Each iteration increases the total weight of $G$ by no more than a multiplicative factor of $\left(1+\eps\right)$ (since the guards whose weight we double have at most an $\eps$ proportion of the total weight).  Therefore after $k$ iterations the weight has increased to at most 
$$
|G|\cdot\left(1+\eps\right)^k ~\leq~
|G|\cdot\exp\left({\frac{k}{2\cdot\opt}}\right) ~\leq~
|G|\cdot2^{\left({\frac{3k}{4\cdot\opt}}\right)}~.
$$
Let $\HHH\subseteq G$ be an optimal hitting set (\ie{} guarding set) of size $\opt$.  For an element $h\in \HHH$ define $z_h$ as the number of times the weight of $h$ has been doubled.  Since $\HHH$ is a hitting set, in each iteration some guard in $\HHH$ has its weight doubled, so we have 
$$
\sum_{h\in \HHH} z_h \geq k
$$
and
\begin{eqnarray*}
w(\HHH) &=& \sum_{h\in\HHH} 2^{z_h}\\
&\geq& \opt\cdot 2^{\left(\frac{k}{\opt}\right)}~~~~~\mbox{(since~} 2^x \mbox{~is a convex function).}
\end{eqnarray*}
We now have
$$
\opt\cdot 2^{\left(\frac{k}{\opt}\right)} ~\leq~
w(\HHH) ~\leq~
w(G) ~\leq~
|G|\cdot2^{\left({\frac{3k}{4\cdot\opt}}\right)}~,
$$
which gives us
$$
k \leq 4\cdot\opt\cdot\log\left(\frac{|G|}{\opt}\right)~.
$$
This bound also tells us that the total weight $w(G)$ never exceeds $\frac{|G|^4}{\opt^3}$.

We must now address the fact that the value of $\opt$ is unknown.  We maintain a variable $c'$ which is our guess at the value of $\opt$, starting with $c'=1$.  If the algorithm runs for more than $4\cdot c'\cdot\log\left(\frac{|G|}{c'}\right)$ iterations without obtaining a guarding set, this implies that there is no guarding set of size $c'$ so we double our guess.  When our algorithm eventually obtains a hitting set, we have $\opt\leq c'\leq 2\cdot \opt$.  The hitting set obtained is a $\left(\frac{1}{2c'}\right)$-net build by our net finder.  Therefore, using the method from Section \ref{sec:hierarchical} to build an $\eps$-net of size $O\kern-2pt\left(\frac{1}{\eps} \log\log \frac{1}{\eps}\right)$, we obtain a guarding set of size $O(\opt\cdot\log\log\opt)$.

\subsubsection{Verification}

The main algorithm requires a verification oracle that, given a set $H$ of guards, either states correctly that $H$ guards $T$ or returns a point $p\in T$ that is not seen by $H$.  We can use the techniques of Bose \etal{} \cite{bose02} to find the visibility polygon of any guard in $H$ efficiently.  It will always be the case that $|H|<n$.  Finding the union of visibility polygons of guards in $H$ can be done in polynomial time, as can comparing this union with $T$.

\hide{
\subsubsection{Fragmentation and the weight function}

Our net finder will subdivide $\partial P$ into a polynomial number of fragments each having equal weight, where a fragment is an interval of $\partial P$ that is the union of line segments (if $G=\partial P$) or vertices (if $G=V(P)$).  To find the endpoints of fragments in polynomial time we need to maintain a concise (\ie{} polynomial) representation of the weight function. 

When $G=V(P)$ we simply maintain the weight function for each vertex $v$.  When $G=\partial P$ we need to maintain the density (\ie{} weight per unit length) for a certain number of segments of the perimeter.  When the algorithm starts, this set of segments is simply the set of edges of the polygon, and each segment has density 1.  In each doubling step, we will double the density for whatever part of $\partial P$ does not see some unguarded point $p$.  For any edge $e$ of $\partial P$ and any point $p\in P$, the points in $e$ that see $p$ form a closed segment, \ie{} a contiguous interval of $e$, that can be computed in polynomial time.  Thus we will double the weight of at most 2 contiguous intervals of each edge.  Since the number of doubling steps is polynomial, we can maintain the weight function as a polynomially sized set of line segments whose union is $\partial P$; furthermore, each segment's density is equal to $2^k$ for some nonnegative integer $k$ that is polynomial in $n$.

}

\section{Building quadratic nets}\label{sec:quadratic}

In this section we show how to build an $\eps$-net using $O(1/\eps^2)$ guards.  This result is not directly useful to us but we use this section to perform the geometric leg work, and hopefully provide some intuition, without worrying about the hierarchical decomposition to be described in Section \ref{sec:hierarchical}.  It should be clear that these $\eps$-nets can be constructed in polynomial time.

\subsection{Subdividing the Perimeter.}

For the construction both of the $\eps$-nets in this section and those in the next section we will subdivide the perimeter into a number of \emph{fragments}.  Fragment endpoints will always lie on vertices, but the weight of a guard location may be split between multiple fragments and a fragment may consist of a single vertex.

The key difference between the construction of the $\eps$-nets in this section and those in the next section is the method of fragmentation.  In this section, the perimeter will simply be divided into $ m = 4/\eps$ fragments each having weight $\frac{\eps}{4}w(G)$.  For our purposes, $1/\eps$ will always be an integer so $m$ will always be an integer.

\subsection{Placing Extremal Guards.}

For two fragments $A_i$ and $A_j$ we will place guards at \emph{extreme points of visibility}.  Those are the first and last points on $A_i$ seen from $A_j$ and the first and last points on $A_j$ seen from $A_i$.  For a contiguous fragment we define the first (resp. last) point of the segment according to the natural clockwise ordering on the perimeter.  We use $G(A_i,A_j)$ to denote the set of up to 4 extremal guards placed between $A_i$ and $A_j$.

These extreme points of visibility might not lie on vertices.  In fact, it is entirely possible that two fragments $A_i$ and $A_j$ see each other even if no vertex of $A_i$ sees $A_j$ and vice versa.  If an extreme point of visibility is not a potential guard location, we will simply not place a guard there.  Our proofs, in particular the proof of Lemma \ref{lem:tangent_helper}, will only require guards on extreme points of visibility that either lie on vertices or on fragment endpoints.

\subsection{All Pairs Extremal Guarding.}

Our aim in this section is to build an $\eps$-net by placing extremal guards for every pair $(A_i,A_j)$ of fragments.  We denote this set of guards with $$S_{AP}=\bigcup_{i\neq j}G(A_i,A_j)~.$$  Note that $|S_{AP}| \leq 4{m\choose 2} = O(1/\eps^2)$.  Also note that every fragment endpoint is included in $S_{AP}$. 

\begin{lemma}\label{lem:AP_fragment_bound}
Any point not guarded by $S_{AP}$ sees at most 4 fragments.
\end{lemma}
\begin{cor}
$S_{AP}$ is an $\eps$-net of size $O(1/\eps^2)$.
\end{cor}

For the proof of Lemma \ref{lem:AP_fragment_bound} we need to present additional properties of the fragments that can be seen by a point.  For a point $x$, the fragments seen by $x$ are ordered clockwise in the order they appear on the boundary of $P$.  We need to consider lines of sight from $x$, and what happens when a transition is made from seeing one fragment $A_i$ to seeing the next fragment $A_j$.  There are three possibilities:
\begin{enumerate}
\vspace{-1mm}\item $j=i+1$ and $x$ sees the guard at the common endpoint of $A_i$ and $A_j$
\vspace{-2mm}\item $A_j$ occludes $A_i$, in which case we say that $x$ has a \emph{left tangent} to $A_j$ (see Figure \ref{fig:left_tangent})
\vspace{-2mm}\item $A_i$ was occluding $A_j$, in which case we say that $x$ has a \emph{right tangent} to $A_i$ (see Figure \ref{fig:right_tangent}).
\end{enumerate}

We say a fragment $A$ \emph{owns} a point $x$ if $x$ sees $A$ in a sector of size at least $\pi$.  We assume any point $x$ is owned by at most one fragment; if $x$ is a fragment endpoint it will itself be a guard, and otherwise if $x$ is owned by two fragments then only those two fragments can see it.

\begin{figure}

\begin{minipage}[b]{0.48\linewidth}
\centering
\includegraphics[scale=.7]{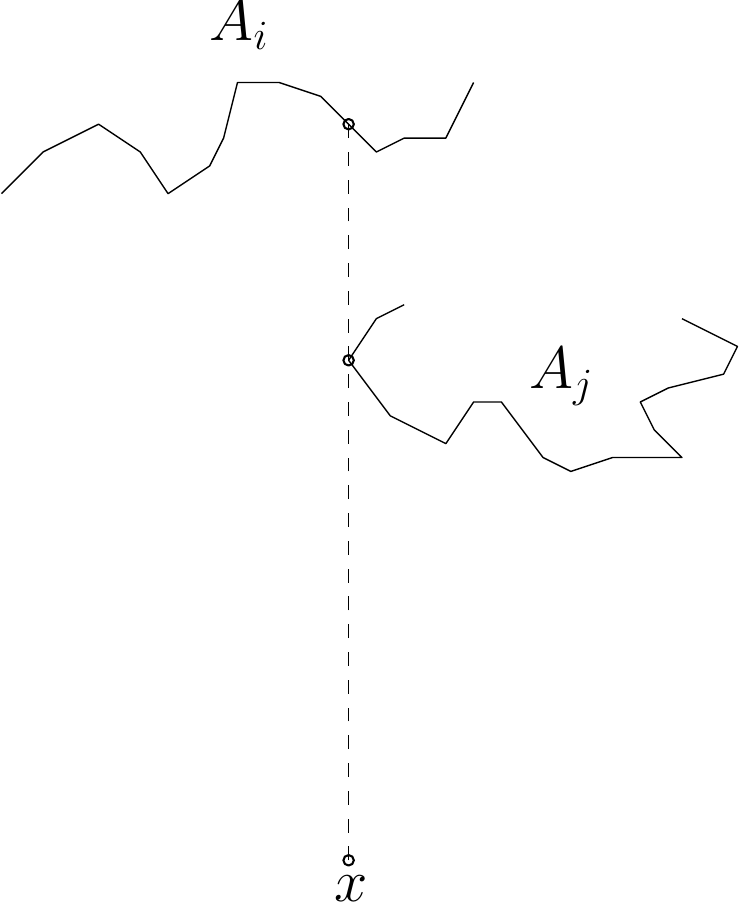}
\caption{\label{fig:left_tangent}The point $x$ has a left tangent to $A_j$.}
\end{minipage}
\hfill
\begin{minipage}[b]{0.48\linewidth}
\centering
\includegraphics[scale=.7]{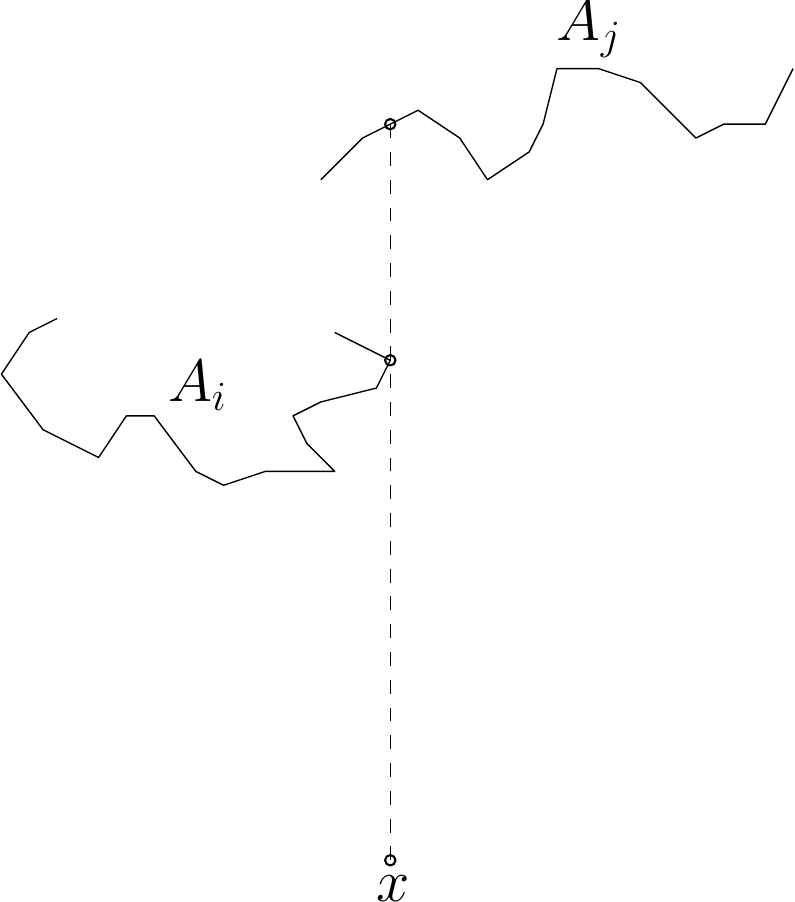}
\caption{\label{fig:right_tangent}The point $x$ has a right tangent to $A_i$.}
\end{minipage}

\end{figure}

\begin{figure}
\centering
\includegraphics[scale=.7]{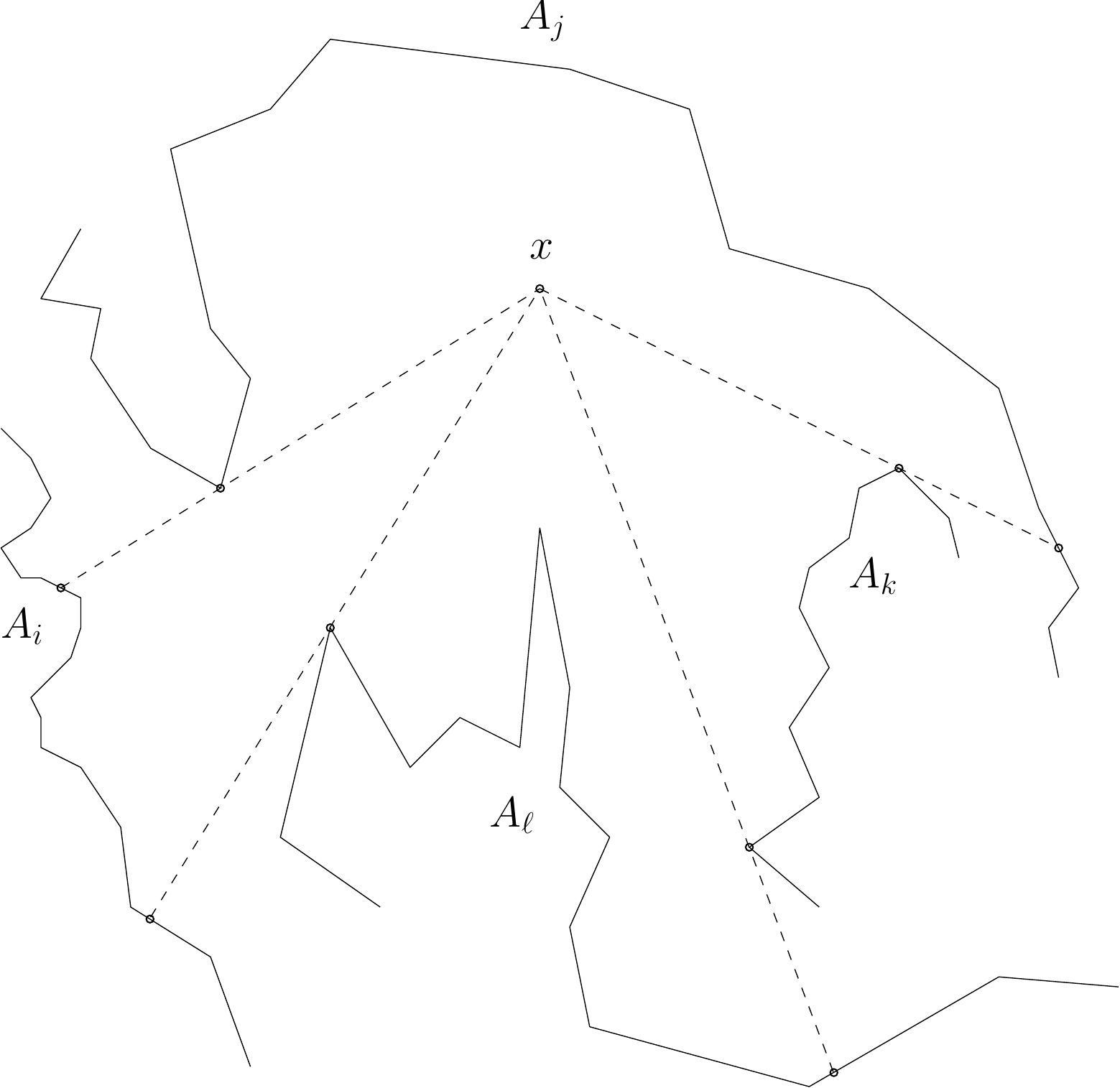}
\caption{\label{fig:cyclic_tangents}The point $x$ has no tangent to $A_i$, a left tangent to $A_j$, both a left and right tangent to $A_k$, and a right tangent to $A_\ell$.  $A_j$ owns $x$.}
\end{figure}

\begin{lemma}\label{lem:tangent_helper}
Let $A_i$, $A_j$, $A_k$ be fragments that are seen by $x$ consecutively in clockwise order.  If $x$ has a left tangent to $A_j$, and the combined angle of $A_j$ and $A_k$ at $x$ is no more than $\pi$, then $x$ sees a guard in $G(A_j,A_k)$.  Symmetrically, if $x$ has a right tangent to $A_j$, and the combined angle of $A_i$ and $A_j$ at $x$ is no more than $\pi$, then $x$ sees a guard in $G(A_i,A_j)$.
\end{lemma}
\begin{proof}
We can assume w.l.o.g.~that $x$ has a left tangent to $A_j$ since the proof of the other case is symmetric.  There are now two cases we have to deal with, depending on whether $x$ has a right tangent to $A_j$ (case 1) or a left tangent to $A_k$ (case 2).  Define $p_L$ and $p_R$ respectively as the first and last points on $A_j$ seen by $x$.  Observe that $x$ must see every vertex on the geodesic between $p_L$ and $p_R$.  Let $q$ be the first point on $A_j$ seen from $A_k$.  In both cases 1 and 2 (see Figures \ref{fig:case_1} and \ref{fig:case_2}), $q$ must be a vertex of the geodesic between $p_L$ and $p_R$.  This can be shown by contradiction; if $q$ lies between consecutive vertices of this geodesic then those two consecutive vertices must also be seen from $A_k$, and one of them comes before $q$.

The restriction that the combined angle of $A_j$ and $A_k$ at $x$ is no more than $\pi$ is necessary to ensure that the geodesic of interest from $A_k$ to $A_j$ does not `pass behind' $x$ to see a point on $A_j$ before $p_L$.

It should be emphasized that, since there is a left tangent to $A_j$, $p_L$ will always be a vertex.  Also, if $p_R$ is not a vertex it will not be the first point on $A_j$ seen from $A_k$.
\end{proof}

\begin{figure}
\begin{minipage}[b]{0.48\linewidth}
\centering
\includegraphics[scale=.7]{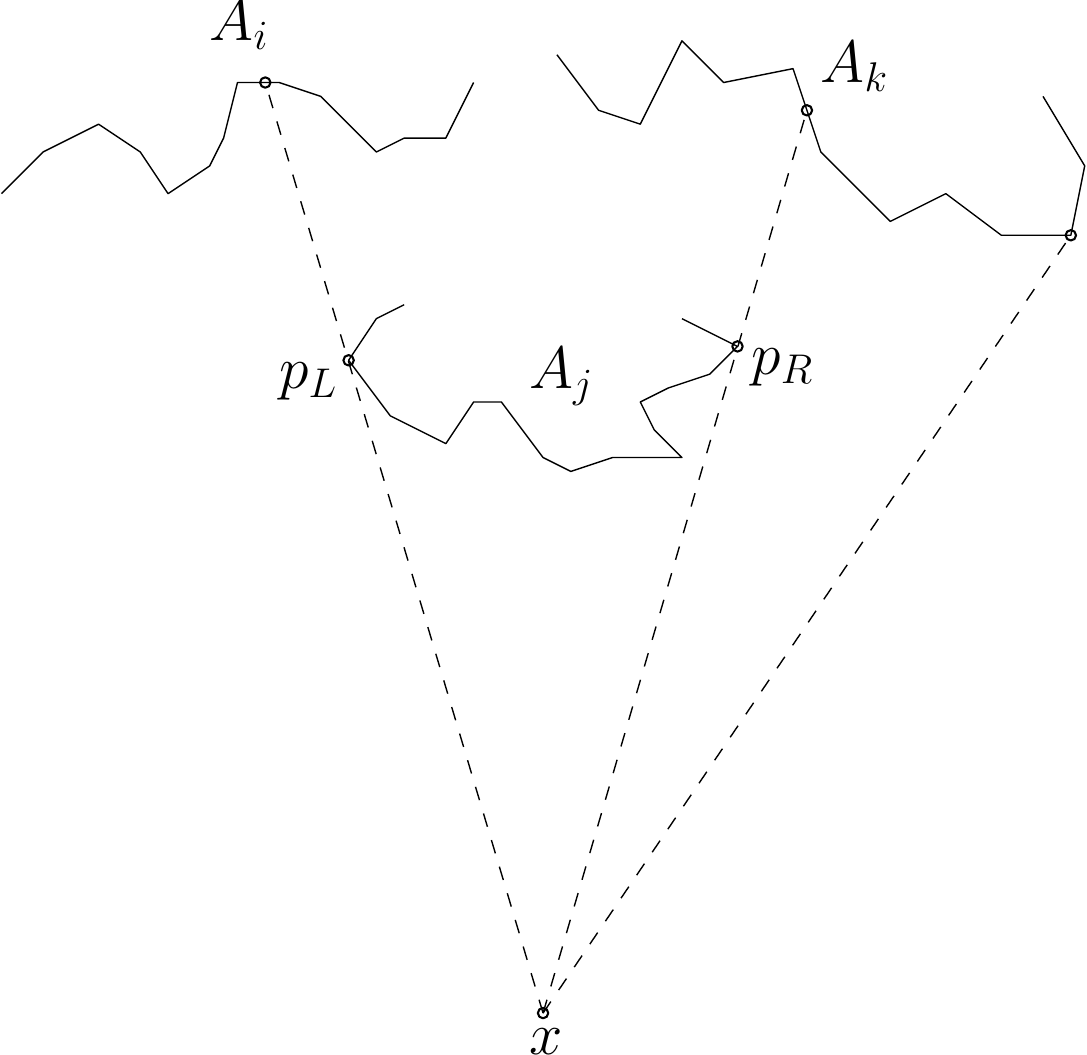}
\caption{\label{fig:case_1}Case 1 in the proof of Lemma \ref{lem:tangent_helper}.  The point $x$ has a left tangent and a right tangent to $A_j$.}
\end{minipage}
\hfill
\begin{minipage}[b]{0.48\linewidth}
\centering
\includegraphics[scale=.7]{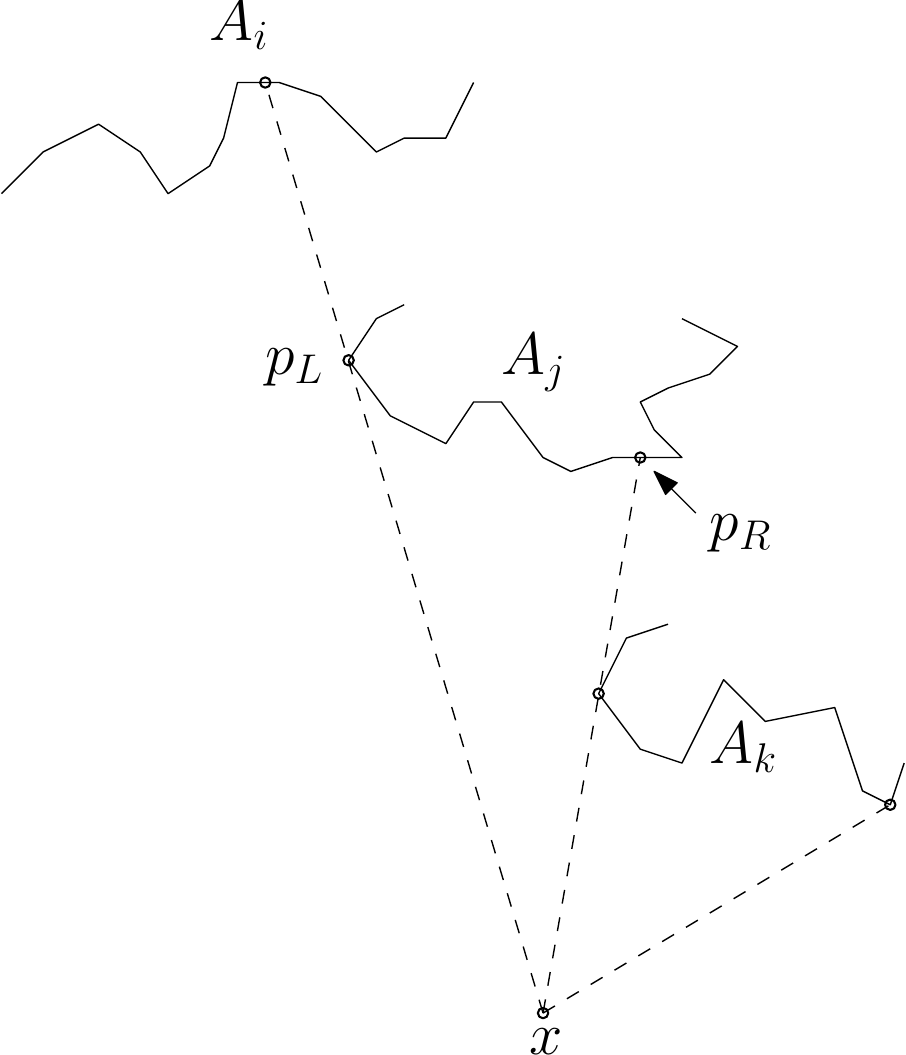}
\caption{\label{fig:case_2}Case 2 in the proof of Lemma \ref{lem:tangent_helper}.  The point $x$ has a left tangent to $A_j$ and a left tangent to $A_k$.}
\end{minipage}

\end{figure}

The proof of Lemma \ref{lem:AP_fragment_bound} is now fairly straightforward.
\begin{proof}[Proof of Lemma \ref{lem:AP_fragment_bound}]
Let $x$ be a point that sees at least 5 fragments.  Assume $x$ is not a fragment endpoint, otherwise it is itself a guard in $S_{AP}$.  If we have a directed graph whose underlying undirected graph is a cycle, then either we have a directed cycle or we have a vertex with in-degree 2.  By the same principle, either some fragment seen by $x$ has no tangent from $x$, or every fragment seen by $x$ has a left tangent from $x$ (or every one has a right tangent, which can be handled symmetrically).

If a fragment seen by $x$ has no tangent from $x$, call such a fragment $A_0$ and let $A_{-2}$, $A_{-1}$, $A_0$, $A_1$, $A_2$ be fragments seen by $x$ in clockwise order.  If the combined angle at $x$ of $A_{-2}$ and $A_{-1}$ is more than $\pi$, the combined angle of $A_1$ and $A_2$ is less than $\pi$.  So we can apply Lemma \ref{lem:tangent_helper} with one of the two pairs of fragments to show that $x$ is seen by a guard.

If every fragment seen by $x$ has a left tangent from $x$, then we can apply Lemma \ref{lem:tangent_helper} using two consecutive fragments with a combined angle at $x$ of less than $\pi$.

\end{proof}

Before we move on we will prove one more helpful lemma.
\begin{lemma}\label{lem:non-tangent}
The number of fragments seen by an unguarded point $x$ that do not have a tangent from $x$ is at most 1.
\end{lemma}
\begin{proof}
Assume the contrary and let $A_0$ and $A_i$ be two such fragments.  If one such fragment owns $x$, assume it is $A_0$ and call the next two fragments seen by $x$ in the clockwise direction $A_1$ and $A_2$ respectively.  By Lemma \ref{lem:tangent_helper}, $x$ is seen by a guard in $G(A_1,A_2)$ so we reach a contradiction.  If no such fragment owns $x$ then assume w.l.o.g. that, over the fragments seen by $x$ between $A_0$ and $A_i$ going clockwise, the combined angle at $x$ is less than $\pi$ (if this is not true it must be true going counterclockwise).  Again, $x$ is seen by a guard in $G(A_1,A_2)$ so we reach a contradiction.
\end{proof}

\section{Hierarchical fragmentation}\label{sec:hierarchical}

In the last section we showed how a quadratic number of guards (\ie{} $O(1/\eps^2)$) could be placed to ensure that any unguarded point sees at most 4 fragments.  In this section we discuss how hierarchical fragmentation can be used to reduce the number of guards required to $O\left(\frac{1}{\eps}\log\log\frac{1}{\eps}\right)$.  We will use $S_{HF}$ to denote the guarding set constructed in this section.  It should be clear that these $\eps$-nets can be constructed in polynomial time.

We can consider the hierarchy as represented by a tree.  At the root there is a single fragment representing the entire perimeter of the polygon.  This root fragment is broken up into a certain number of child fragments.  Fragmentation continues recursively until a specified depth $t$ is reached.  We will set $t=\ceil{\log\log\frac{1}{\eps}}$.  The \emph{fragmentation factor} (equivalently, the branching factor of the corresponding tree) is not constant, but rather depends on both $t$ and the level in the hierarchy.  The fragmentation factor generally decreases as the level of the tree increases.  Specifically, if $b_i$ is the fragmentation factor at the $i^{\rm th}$ step, we have
$$
b_i = \left\{
\begin{array}{lcl}
2^{2^{t-1}+1}\cdot 4t\cdot 2^{1-t} \cdot \alpha&,&i=1\\
2^{2^{t-i}+1}&,&1<i\leq t~,
\end{array}
\right.
$$
where $\alpha \leq 1$ is a term introduced only to deal with an issue arising from ceilings and double exponentials, namely the fact that $2^{2^{\ceil{\log\log1/\eps}}}$ is not in $O(1/\eps)$.  $\alpha$ is specified in (\ref{eqn:alpha}) later in this section.

If $f_i$ is the total number of fragments after the $i^{\rm th}$ fragmentation step, this gives us
$$
f_i = \left\{
\begin{array}{ccl}
1 &,&i=0\\
4t\cdot 2^{2^t - 2^{t-i} -t +i +1} \cdot \alpha&,&0<i\leq t \\
4t\cdot 2^{2^t} \cdot \alpha&,&i = t~,
\end{array}
\right.
$$
since
\begin{eqnarray*}
f_i &=& \prod_{j=1}^ib_j \\
&=& 4t\cdot 2^{1-t}\cdot \alpha \cdot \prod_{j=1}^i 2^{2^{t-j}+1} \\
&=& 4t\cdot 2^{1-t+\sum_{j=1}^i (2^{t-j}+1)} \cdot \alpha \\
&=& 4t\cdot 2^{2^t - 2^{t-i} -t +i +1} \cdot \alpha~~.
\end{eqnarray*}

Our algorithm will place guards at all pairs of \emph{sibling fragments}, \ie{} fragments having the same parent fragment.  For the purposes of this guard placement, the complement of the parent fragment, \ie{} the subset of $G$ outside the parent fragment, will be considered a dummy child fragment.  That is, it will be considered a child fragment when placing guards, but not when counting the number of child fragments seen from some point $x$ as in the statement of Corollary \ref{cor:children} or in the proof of Lemma \ref{lem:HF_fragment_bound}.  To denote the complement of a fragment $A$ we use $\overline{A}$.  Considering $\overline{A}$ to be a child of $A$ when placing guards allows us to consider the children of $A$ as if they were fragments with guards placed for all pairs.  For example, we can obtain the following corollary from Lemmas \ref{lem:AP_fragment_bound} and \ref{lem:non-tangent}.

\begin{cor}\label{cor:children}
For an unguarded point $x$ and a fragment $A$, the number of child fragments of $A$ seen by $x$ is at most 3, and at most one of these child fragments does not have a tangent from $x$.
\end{cor}

The total number of guards placed will be
$$
|S_{HF}| 
~~\leq~~ 4 \sum_{i=1}^t {b_i+1 \choose 2} f_{i-1} 
~~\leq~~ 4 \sum_{i=1}^t b_i^2 f_{i-1}~~.
$$
If $t\geq 6$ we have $b_i\leq 2^{2^{t-i}+1}$ for all values of $i$.  This gives us
\begin{eqnarray*}
|S_{HF}| &\leq& 4 \alpha \sum_{i=1}^t 2^{2\left(2^{t-i}+1\right)} \cdot 4t\cdot 2^{2^t - 2^{t-i+1} -t +i} \\
&=& 16t \alpha \sum_{i=1}^t 2^{2^t -t +i +2} \\
&=& 16t \alpha\cdot 2^{2^t-t+3} (2^t-1) \\
&<& 16t \alpha\cdot 2^{2^t+3} \\
&=& 128t \alpha\cdot 2^{2^t}~.
\end{eqnarray*}

Recall that $t = \ceil{\log \log \frac 1 \eps}$.  We need to define $\alpha$ in a way that ensures $b_1$ is an integer and that ensures the following two equations hold:

\begin{equation}\label{eqn:HF_num_guards}
|S_{HF}| ~=~ O\kern-2pt\left(\frac{1}{\eps} \log \log \frac{1}{\eps}\right)
\end{equation}
\begin{equation}\label{eqn:HF_small_fragments}
\frac{f_t}{4t} ~\geq~ \frac{1}{\eps}~.
\end{equation}

To satisfy these three criteria, it suffices to set
\begin{equation}\label{eqn:alpha}
\alpha 
~=~ \frac{\ceil{ 2^{2^{t-1}+1} \cdot 4t \cdot 2^{-t} \cdot 2^{\log(1/\eps) - 2^t} }}{ 2^{2^{t-1}+1}\cdot 4t\cdot 2^{-t} }
~=~ \frac{\ceil{ 4t\cdot 2^{\log(1/\eps)+1-t-2^{t-1}}}}{ 4t\cdot 2^{2^{t-1}+1-t} }~.
\end{equation}

We must now provide a generalization of Lemma \ref{lem:AP_fragment_bound} that works with our hierarchical fragmentation.
\begin{lemma}\label{lem:HF_fragment_bound}
Any point not guarded by $S_{HF}$ sees at most $4i$ fragments at level $i$.
\end{lemma}
\noindent Applying this with $i=t$ and using (\ref{eqn:HF_num_guards}) and (\ref{eqn:HF_small_fragments}), we get
\begin{cor}
$S_{HF}$ is an $\eps$-net of size $O\kern-2pt\left(\frac{1}{\eps} \log \log \frac{1}{\eps}\right)$.
\end{cor}
\begin{proof}[Proof of Lemma \ref{lem:HF_fragment_bound}]
Let $x$ be a point that does not see any guard in $S_{HF}$.  From the tree associated with the hierarchical fragmentation, we consider the subtree of fragments that see $x$.  We define a \emph{branching fragment} as a fragment with multiple children seen by $x$ and we claim that at any level there are at most 2 branching fragments.  Corollary \ref{cor:children} tells us that any fragment has at most 3 children seen by $x$.  At level 1 there are at most 4 fragments seen by $x$, so it follows that the number of fragments seen by $x$ at level $i$ is at most $4i$.  We must now prove our claim that there are at most 2 branching fragments at any level.

First we note that a branching fragment either has no tangent from $x$ or owns $x$.  To see this, consider a fragment $A$ that has a tangent from $x$ and does not own $x$.  Assume w.l.o.g.\null{} that $x$ has a left tangent to $A$ and call the point of tangency $p_L$.  $x$ must then also have a left tangent to the child fragment $A_0$ of $A$ that contains $p_L$.  $A_0$ must be the leftmost child fragment of $A$ seen by $x$.  If $x$ sees another child fragment $A_1$ of $A$ to the right of $A_0$, then by Lemma \ref{lem:tangent_helper} it is seen by a guard in $G(A_0,A_1)$.

Consider now the following possibilities for a given fragment $A$.
\begin{enumerate}
\item $A$ is not seen by $x$.  Clearly $x$ cannot see any child fragments of $A$.

\item $A$ does not own $x$, and $x$ has a tangent to $A$.  $A$ then has exactly one child fragment that sees $x$, and this fragment is of type (2).

\item $A$ does not own $x$, and $x$ does not have a tangent to $A$.  By Corollary \ref{cor:children}, $x$ can see at most 3 child fragments of $A$.  At most one of these children is of type (3) and all others must be of type (1) or (2).

\item $A$ owns $x$ and has no tangents from $x$, \ie{} $\overline{A}$ has two tangents from $x$.  If a child of $A$ owns $x$ it must be the only child of $A$ that sees $x$, and this child is also of type (4).  Otherwise, $A$ would have a child fragment $A_i$ that is seen by $x$, does not own $x$, and is adjacent to $\overline{A}$.  $x$ would then be seen by a guard in $G(A_i,\overline{A_P})$.  Thus $A$ has at most one child that is not of type (1) or (2).

\item $A$ owns $x$ and has two tangents from $x$.  Because $\overline{A}$ is, in a sense, a `dummy' child of type (3), $A$ cannot have a real child of type (3) by the proof of Lemma \ref{lem:non-tangent}.  Further, if $A$ has a child $A_0$ that owns $x$, this child must also be of type (5).  Otherwise assume w.l.o.g. that $A_1$, immediately clockwise from $A_0$, has a left tangent from $x$.  Then, using $A_2$ to denote the fragment clockwise from $A_1$ ($A_2$ might be $\overline{A_P}$), $x$ is seen by $G(A_1,A_2)$.

\item $A$ owns $x$ and has exactly one tangent from $x$ (see Figure \ref{fig:fruitful_case_6}).  We consider how $A$ can have multiple children seen by $x$.  Assume w.l.o.g. that $\overline{A}$ has a right tangent.  If $A_{-1}$ is the child of $A$ seen by $x$ immediately counterclockwise from $\overline{A}$ then $A_{-1}$ must own $x$, otherwise $x$ is seen by $G(A_{-1},\overline{A})$.  If $A_1$ is the child of $A$ seen by $x$ immediately clockwise from $\overline{A}$ then $A_1$ cannot have a tangent from $x$ otherwise $x$ would be seen by $G(\overline{A}, A_1)$.  If $x$ can see $A_2$, a child of $A$ between $A_1$ and $A_{-1}$, then $x$ must have two tangents to $A_{-1}$ otherwise it would be seen by $G(A_1,A_2)$.

Therefore if $A$ has more than one child seen by $x$, there must one of type (3) and one of type (5), plus (possibly) a child of type (2).

\end{enumerate}

\begin{figure}
\centering
\includegraphics[scale=.6]{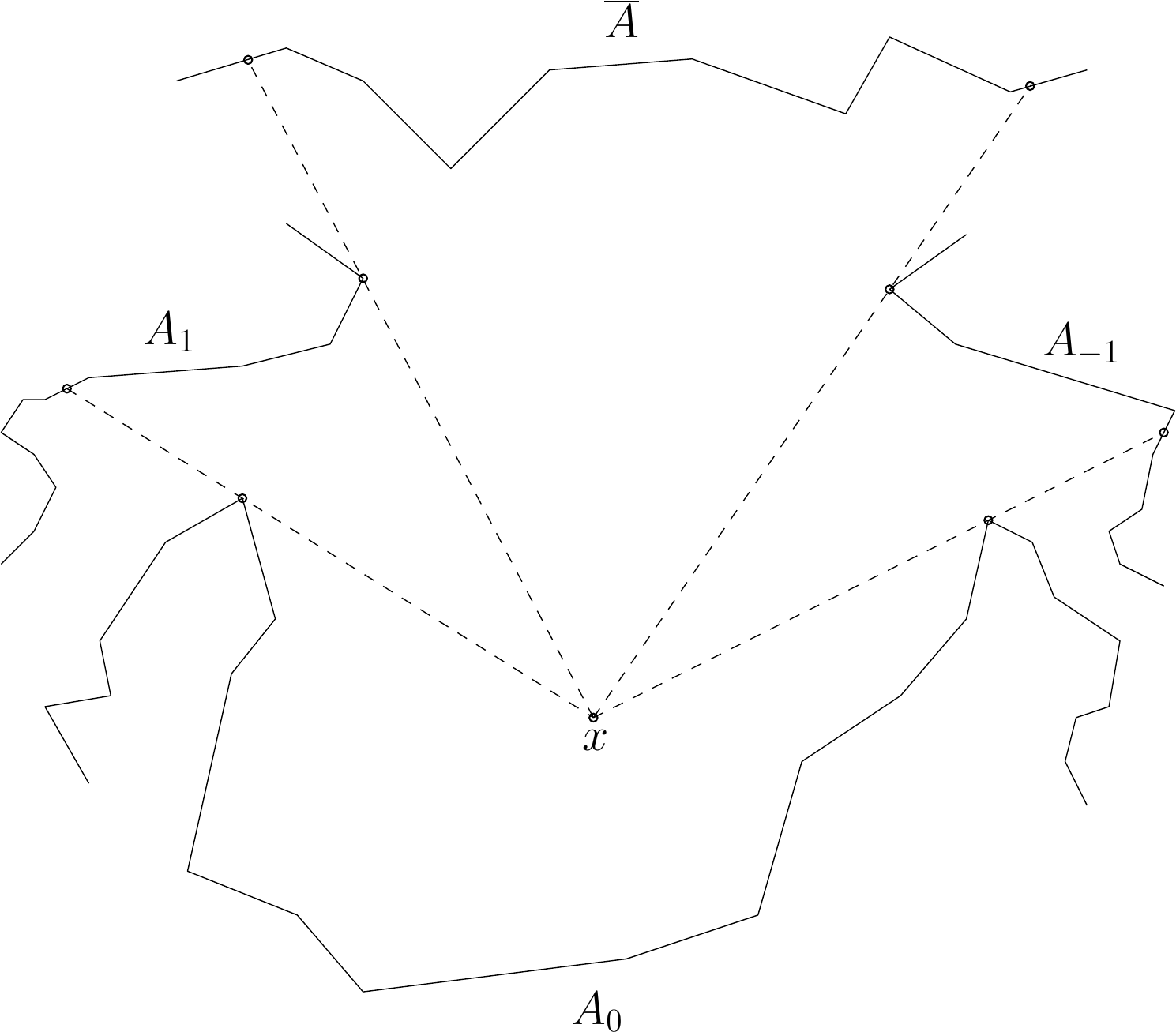}
\caption{\label{fig:fruitful_case_5}The only way a fragment of type (5) can have three children seen by $x$.}
\end{figure}

\begin{figure}
\centering
\includegraphics[scale=.6]{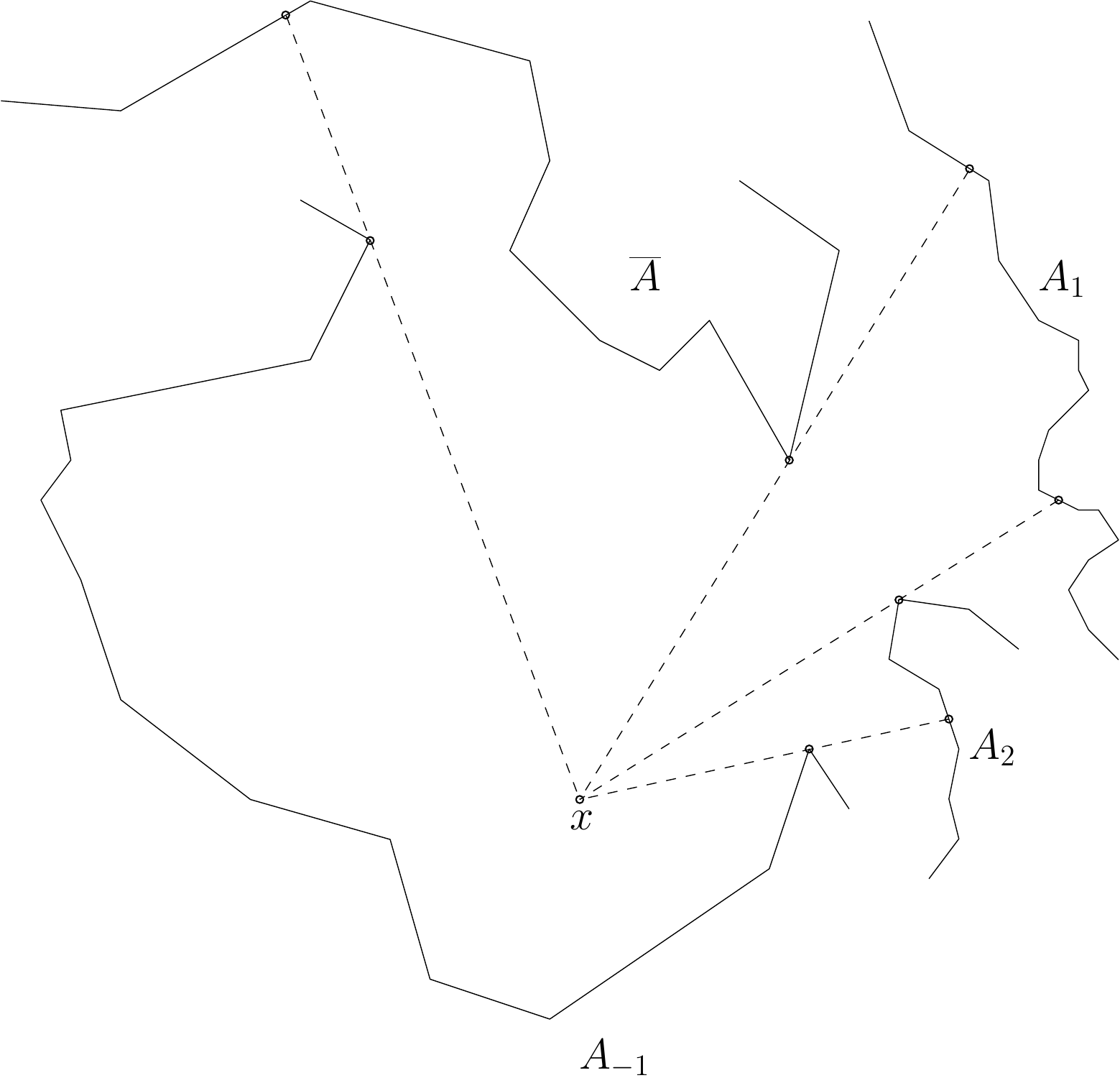}
\caption{\label{fig:fruitful_case_6}The only way a fragment of type (6) can have three children seen by $x$.}
\end{figure}

We call a non-root fragment \emph{fruitful} if it or one of its descendants is a branching fragment.  Only fragments of type (3-6) can be fruitful.  Only fragments of type (6) can have more than one fruitful child, and they can have at most two fruitful children.  No non-root fragment can have a child fragment of type (6).  Also, if the root has a child fragment of type (6), the root cannot have a child of type (3).  Therefore any level has at most 2 fruitful fragments.

We can now state the following:
\begin{itemize}
\item Level 1 has at most 4 child fragments that see $x$, at most 2 of which are fruitful.
\item A fruitful fragment has at most 3 child fragments that see $x$, at most 1 of which is fruitful.
\item A non-fruitful fragment has at most 1 child fragment that sees $x$.
\end{itemize}

Therefore any level has at most 2 fruitful fragments and the number of fragments at level $i$ that see $x$ is at most $4i$.

\end{proof}

\section{Open problems}

\begin{itemize}
\item We have obtained a $o(\log\opt)$-approximation factor for vertex guards and perimeter guards.  Can the same be done for point guards?

\item Can we do better than $O(\log\log\opt)$ for perimeter guards?  In particular, can we find a constant factor approximation algorithm to match the hardness of approximation result of Eidenbenz \cite{eidenbenz98}?  

\item For simple polygons, the set systems associated with point guards have maximum VC-dimension at least 6 and at most 23 \cite{valtr98}; it is believed that the true value is closer to the lower end of this range, perhaps even 6 \cite{kalai97}.  The upper bound of 23 holds \textit{a fortiori} for set systems associated with perimeter guards but the lower bound of 6 does not.  A lower bound of 4 follows from a trivial modification to an example for monotone chains \cite{king08}; we can increase this bound to 5 without too much difficulty (see Figure \ref{fig:vc_dim_5}).  Can set systems associated with perimeter guards actually have VC-dimension as high as 6?  And can the upper bound of 23 be improved?  It seems that improving the upper bound would be easier for perimeter guards than for point guards.
\end{itemize}

\begin{figure}
\centering
\includegraphics[scale=1]{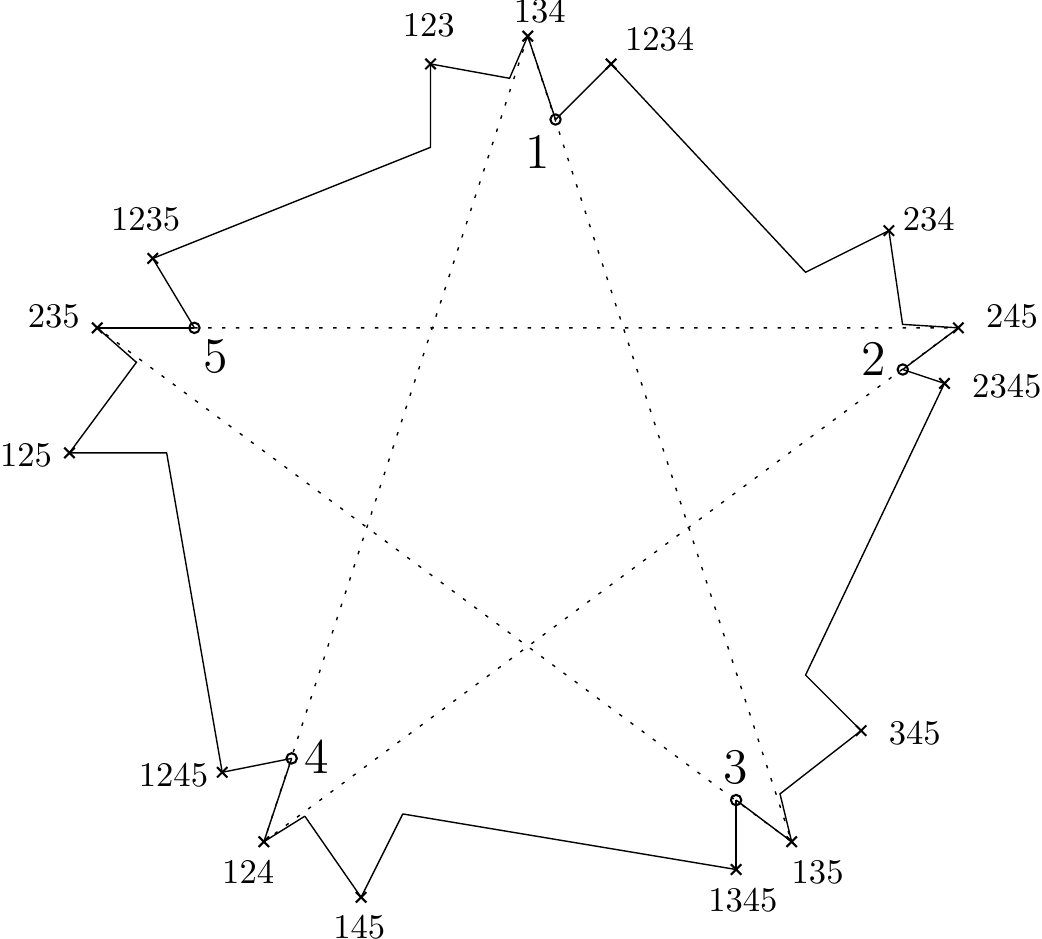}
\caption{\label{fig:vc_dim_5}A polygon with a set $S$ of 5 points on the perimeter.  The points in $S=\{1,2,3,4,5\}$ are marked with circles and labeled with large numbers.  Each point in $S$ sees all of $S$, and each guard seeing a subset of $S$ of size 3 or 4 is marked with a cross and labeled with small numbers indicating which points in $S$ it sees.  Guards seeing the 16 subsets of $S$ of size 0, 1, or 2 are not shown.  Adding these is a simple matter of adding nooks with very small angles of visibility, thus we can construct a polygon with 5 points on the perimeter shattered by $2^5$ perimeter guards.  Such a polygon can also be obtained via a fairly straightforward modification of the example of Kalai and Matou\v{s}ek for point guards \cite{kalai97}.}
\end{figure}

\bibliographystyle{plain}
\bibliography{art_gallery_epsilon-nets}

\begin{thebibliography}{10}

\bibitem{aggarwal84}
Alok Aggarwal.
\newblock {\em The art gallery theorem: its variations, applications and
  algorithmic aspects}.
\newblock PhD thesis, The Johns Hopkins University, 1984.

\bibitem{aronov09}
B.~Aronov, E.~Ezra, and M.~Shair.
\newblock {Small-size $\varepsilon$-nets for axis-parallel rectangles and
  boxes}.
\newblock In {\em Proceedings of the 41st annual ACM symposium on Theory of
  computing}, pages 639--648. ACM, 2009.

\bibitem{blumer89}
Anselm Blumer, A.~Ehrenfeucht, David Haussler, and Manfred~K. Warmuth.
\newblock {Learnability and the Vapnik-Chervonenkis dimension}.
\newblock {\em J. ACM}, 36(4):929--965, 1989.

\bibitem{bose02}
Prosenjit Bose, Anna Lubiw, and J.~Ian Munro.
\newblock Efficient visibility queries in simple polygons.
\newblock {\em Comput. Geom. Theory Appl.}, 23(3):313--335, 2002.

\bibitem{Bronnimann95}
H.~Br{\"o}nnimann and M.~T. Goodrich.
\newblock Almost optimal set covers in finite {VC}-dimension.
\newblock {\em Discrete \& Computational Geometry}, 14(1):463--479, 1995.

\bibitem{chvatal79}
V.~Chv\'{a}tal.
\newblock {A greedy heuristic for the set-covering problem}.
\newblock {\em Mathematics of Operations Research}, 4(3):233--235, 1979.

\bibitem{deshpande07}
Ajay Deshpande, Taejung Kim, Erik~D. Demaine, and Sanjay~E. Sarma.
\newblock A pseudopolynomial time ${O}(\log n)$-approximation algorithm for art
  gallery problems.
\newblock In {\em Proceedings of the 10th Workshop on Algorithms and Data
  Structures (WADS 2007)}, volume 4619 of {\em Lecture Notes in Computer
  Science}, pages 163--174, Halifax, Nova Scotia, Canada, August 15--17 2007.

\bibitem{efrat06}
Alon Efrat and Sariel Har-Peled.
\newblock Guarding galleries and terrains.
\newblock {\em Inf. Process. Lett.}, 100(6):238--245, 2006.

\bibitem{eidenbenz98}
S.~Eidenbenz.
\newblock {Inapproximability results for guarding polygons without holes}.
\newblock {\em Lecture notes in Computer Science}, 1533:427--436, 1998.

\bibitem{Eidenbenz01}
S.~Eidenbenz, C.~Stamm, and P.~Widmayer.
\newblock Inapproximability results for guarding polygons and terrains.
\newblock {\em Algorithmica}, 31(1):79--113, 2001.

\bibitem{Feige98}
U.~Feige.
\newblock A threshold of $\ln n$ for approximating set cover.
\newblock {\em Journal of the ACM}, 45(4):634--652, July 1998.

\bibitem{Garey79}
M.~Garey and D.~Johnson.
\newblock {\em Computers and Intractibility: {A} Guide to the Theory of
  {NP}-Completeness}.
\newblock W.H. Freeman and Co., 1979.

\bibitem{ghosh87}
S.~Ghosh.
\newblock Approximation algorithms for art gallery problems.
\newblock In {\em Proc. Canadian Information Processing Society Congress},
  pages 429--434, 1987.

\bibitem{kalai97}
G.~Kalai and J.~Matou\v{s}ek.
\newblock Guarding galleries where every point sees a large area.
\newblock {\em Israel Journal of Math}, 101(1):125--139, 1997.

\bibitem{king08}
J.~King.
\newblock {VC}-dimension of visibility on terrains.
\newblock In {\em Proceedings of CCCG 2008}, pages 27--30, 2008.

\bibitem{komlos92}
J.~Koml{\'o}s, J.~Pach, and G.~Woeginger.
\newblock {Almost tight bounds for $\varepsilon$-Nets}.
\newblock {\em Discrete and Computational Geometry}, 7(1):163--173, 1992.

\bibitem{lee86}
D.~Lee and A.~Lin.
\newblock {Computational complexity of art gallery problems}.
\newblock {\em IEEE Transactions on Information Theory}, 32(2):276--282, 1986.

\bibitem{orourke87}
Joseph O'Rourke.
\newblock {\em Art Gallery Theorems and Algorithms}.
\newblock Oxford University Press, 1987.
\newblock
  \url{http://maven.smith.edu/~orourke/books/ArtGalleryTheorems/art.html}.

\bibitem{orourke83}
Joseph O'Rourke and Kenneth~J. Supowit.
\newblock {Some NP-hard polygon decomposition problems}.
\newblock {\em IEEE Transactions on Information Theory}, 29(2):181--189, 1983.

\bibitem{Raz97}
Ran Raz and Shmuel Safra.
\newblock A sub-constant error-probability low-degree-test and a sub-constant
  error-probability pcp characterization of np.
\newblock In {\em In Proc. 29th ACM Symp. on Theory of Computing, 475-484. El
  Paso}, 1997.

\bibitem{valtr98}
P.~Valtr.
\newblock Guarding galleries where no point sees a small area.
\newblock {\em Israel Journal of Mathematics}, 104(1):1--16, 1998.

\end{thebibliography}

\end{document}